\documentclass[11pt]{article}
\usepackage[a4paper, margin=1in]{geometry}
\usepackage{amsmath, amssymb, amsthm}
\usepackage{graphicx}
\usepackage{hyperref}
\usepackage{titling}
\usepackage{tocloft}
\usepackage{xcolor}
\usepackage{slashed}
\usepackage{geometry}
\usepackage{longtable}
\usepackage{caption}
\usepackage{tabularx}
\usepackage{mathtools}
\usepackage{algorithm}
\usepackage{algorithmic}
\usepackage[all,cmtip]{xy} 

\newtheorem{theorem}{Theorem}
\newtheorem{lemma}{Lemma}
\newtheorem{corollary}{Corollary}
\newtheorem{definition}{Definition}
\newtheorem{conjecture}{Conjecture}
\newtheorem{proposition}{Proposition}

\newtheorem{remark}{Remark}

\hypersetup{
    colorlinks=true,
    linkcolor=blue,
    filecolor=magenta,
    urlcolor=blue,
    citecolor=blue,
    pdftitle={On the Holographic Geometry of Deterministic Computation},
    pdfpagemode=FullScreen,
}

\title{\Large \textbf{On the Holographic Geometry of Deterministic Computation}\\[0.5in]}
\author{\Large Logan Nye, MD\\[0.3in]
Carnegie Mellon University School of Computer Science\\[0.1cm]
5000 Forbes Ave Pittsburgh, PA 15213 USA\\[0.5cm]
\texttt{lnye@andrew.cmu.edu}\\[0.2cm]
\small ORCID: \href{https://orcid.org/0009-0002-9136-045X}{0009-0002-9136-045X}}
\date{}

\begin{document}

\maketitle
\vspace{0.8in}

\begin{center}
    \Large \textbf{Abstract}
\end{center}
\vspace{0.3in}
Standard simulations of Turing machines suggest a linear relationship between the temporal duration $t$ of a run and the amount of information that must be stored by known simulations to certify, verify, or regenerate the configuration at time $t$. For deterministic multitape Turing machines over a fixed finite alphabet, this apparent linear dependence is not intrinsic: any length-$t$ run can be simulated using $O(\sqrt{t})$ work-tape cells via a Height Compression Theorem for succinct computation trees together with an Algebraic Replay Engine. In this paper we recast that construction in geometric and information-theoretic language. We interpret the execution trace as a spacetime DAG of local update events and exhibit a family of recursively defined \emph{holographic boundary summaries} such that, along the square-root-space simulation, the total description length of all boundary data stored at any time is $O(\sqrt{t})$. Using Kolmogorov complexity, we prove that every internal configuration has \emph{constant} conditional description complexity given the appropriate boundary summary and time index, establishing that the spacetime bulk carries no additional algorithmic information beyond its boundary. We express this as a one-dimensional \emph{computational area law}: there exists a simulation in which the information capacity of the active ``holographic screen'' needed to generate a spacetime region of volume proportional to $t$ is bounded by $O(\sqrt{t})$. In this precise sense, deterministic computation on a one-dimensional work tape admits a holographic representation, with the bulk history algebraically determined by data residing on a lower-dimensional boundary screen.

\newpage

\section{Introduction}
\label{sec:intro}

The tradeoff between deterministic time and space on Turing machines is a foundational theme in complexity theory, with a long history of upper and lower bounds relating these resources~\cite{HopcroftPaulValiant1977,PaulPippengerSzemerediTrotter1983,AroraBarak2009}. Together with classical simulations such as Savitch's theorem relating nondeterministic and deterministic space~\cite{Savitch1970}, these results establish that many natural problems admit nontrivial transformations between time and space bounds. More recently, for deterministic multitape Turing machines it has been shown that a time-$t$ computation can be simulated in space $O\!\bigl(\sqrt{t \log t}\bigr)$~\cite{Williams2025}, sharpening the classical $O(t / \log t)$ bound.

In a companion work (henceforth \emph{Height Compression}) we prove a stronger square-root space simulation theorem
\[
\mathrm{TIME}[t] \;\subseteq\; \mathrm{SPACE}\bigl(O(\sqrt{t})\bigr)
\]
for deterministic multitape Turing machines, where space is measured in tape cells over a fixed finite alphabet~\cite{HeightCompression}. The proof proceeds via three ingredients: a \emph{Height Compression Theorem} for succinct computation trees, an \emph{Algebraic Replay Engine} that regenerates local configurations from short summaries, and a \emph{rolling boundary} discipline for traversing the compressed tree. Informally, the companion paper shows that the entire computation history of a deterministic run can be regenerated from a carefully chosen sequence of low-dimensional \emph{boundary summaries}, using only $O(\sqrt{t})$ active tape cells at any moment.

The present paper develops a geometric and information-theoretic interpretation of this phenomenon. We treat the execution of a deterministic machine not only as a linear sequence of configurations, but as a finite directed acyclic graph (DAG) that we regard as a discrete \emph{spacetime} object, equipped with a combinatorial notion of locality induced by the machine's transition rules. We then reinterpret the technical machinery of height compression in geometric language:

\begin{itemize}
    \item We repackage the \emph{interval summaries} of height compression as \emph{holographic boundary states} that encode all information flowing into and out of a spacetime sub-region.
    \item We reinterpret the height-compressed computation tree as a \emph{static causal tree} of spacetime volumes, on which linear time appears as a particular depth-first traversal.
    \item We define an \emph{active holographic screen} consisting of the boundary states and local replay window maintained by the simulator, and we show that its size satisfies a one-dimensional area law: the maximum screen size over the run is $O(\sqrt{t})$.
\end{itemize}

On top of these formal correspondences, we prove a precise information-theoretic statement: the \emph{bulk} of the spacetime object has $O(1)$ conditional Kolmogorov complexity relative to its boundary summaries. That is, once the machine and the boundary data of a block are fixed, any internal configuration can be produced by a fixed, constant-complexity procedure. In this sense, the interior of a deterministic computation history is an \emph{information-theoretic vacuum}: all nontrivial information resides in its boundary data, and the bulk is an algebraically determined evaluation trace.

Finally, we formulate conjectural extensions of the area law to higher-dimensional memory architectures, and we discuss analogies with holography in quantum gravity and with area laws for entanglement in many-body systems~\cite{Bekenstein1973,tHooft1993,Susskind1995,RyuTakayanagi2006,EisertCramerPlenio2010}. We distinguish carefully between rigorous theorems (which are purely combinatorial and information-theoretic) and speculative analogies.

\medskip

\section{Preliminaries from Height Compression}
\label{sec:prelims}

We briefly recall the core constructions from the height-compression technique, specializing to the aspects needed for the present work. For background on standard models of computation and complexity-theoretic notation, see, e.g.,~\cite{AroraBarak2009}. For Kolmogorov complexity and encoding conventions, we follow~\cite{LiVitanyi2008} and make our choices explicit in Appendix~\ref{app:encodings}.

\subsection{Machine model and spacetime diagram}
\label{subsec:machine-model}

We fix a deterministic multitape Turing machine
\[
M = (Q, \Sigma, \Gamma, \delta, q_0, q_{\mathrm{acc}}, q_{\mathrm{rej}}),
\]
with finite state set $Q$, input alphabet $\Sigma$, work alphabet $\Gamma \supseteq \Sigma$, and transition function
\[
\delta : Q \times \Gamma^k \to Q \times \Gamma^k \times \{-1,0,+1\}^k
\]
for some fixed number of tapes $k \ge 1$. We assume the standard \emph{Lipschitz locality} of head motion: in one time step, each head moves by at most one cell.

A run of $M$ of length $t$ is a sequence of configurations
\[
C_0, C_1, \dots, C_t,
\]
where each $C_\tau$ encodes the tape contents, head positions, and control state at time $\tau$. The \emph{spacetime diagram} of the run is the finite directed acyclic graph
\[
\mathcal{M}_{\mathrm{raw}} = (V,E),
\]
where $V$ consists of all local degrees of freedom (e.g., tape cells with time labels and the control state) and $E$ consists of directed edges representing the causal dependencies induced by $\delta$ between consecutive configurations. For standard one-dimensional tapes we have $V \subseteq \mathbb{Z}^{1+1}$; more generally one may consider higher-dimensional tape lattices.

We will colloquially refer to $\mathcal{M}_{\mathrm{raw}}$ as a ``spacetime manifold'', but for our purposes it is simply a finite DAG equipped with the adjacency relation defined by $\delta$.

\subsection{Block-respecting runs and time-blocks}
\label{subsec:block-respecting}

Fix a \emph{block size} parameter $b \in \mathbb{N}$. We partition the time indices $\{1,\dots,t\}$ into $T = \lceil t/b \rceil$ contiguous \emph{time-blocks}
\[
I_k = \bigl[(k-1)b + 1,\, \min\{kb,t\}\bigr]
\quad \text{for } k = 1, \dots, T.
\]

We will use the following notion of block locality, which is the one realized by the height-compression construction in~\cite{HeightCompression}.

\begin{definition}[Block-respecting run]
\label{def:block-respecting}
Let $b \in \mathbb{N}$ and let $c_{\mathrm{int}} \ge 1$ be a fixed constant. A length-$t$ run of $M$ is \emph{block-respecting with parameters $(b,c_{\mathrm{int}})$} if, for every time-block $I_k$,
all interactions between the computation inside $I_k$ and the rest of the run occur through an \emph{interface window} of size at most $c_{\mathrm{int}} \cdot b$ on the work tape, located in the spacelike slices at the temporal boundaries of $I_k$.

Concretely, the interface for $I_k$ consists of:
\begin{itemize}
    \item the machine state and head positions at the entry and exit times of $I_k$, and
    \item the contents of an interval of at most $c_{\mathrm{int}} \cdot b$ tape cells containing all cells visited by the heads during $I_k$.
\end{itemize}
\end{definition}

The Height Compression paper~\cite{HeightCompression} shows that for suitably chosen $b$ (as a function of $t$) and after an appropriate preprocessing step, one may assume without loss of generality that a long run is block-respecting in the above sense. We treat that reduction as a black box in this paper for brevity (the full proof details are provided in the companion work). This black-box reduction ensures that the heads do not stray too far within a block and that any long excursions are folded into a structured pattern which can be summarized via a fixed-size interface window.

\subsection{Interval summaries and merge operator}
\label{subsec:summaries}

For any time interval $[L,R]$ that is a union of whole blocks, the height-compression construction of~\cite{HeightCompression} associates a constant-size summary that records all information needed to \emph{interface} that interval with its complement.

\begin{definition}[Interval summary]
\label{def:interval-summary}
Let $[L,R]$ be a time interval that is a union of whole blocks $I_k$. The \emph{interval summary} $\sigma([L,R])$ is a finite record of the form
\[
\sigma([L,R]) = \bigl(q_{\mathrm{in}}, q_{\mathrm{out}},
\vec{h}_{\mathrm{in}}, \vec{h}_{\mathrm{out}},
W_{\mathrm{interface}}\bigr),
\]
where:
\begin{itemize}
    \item $q_{\mathrm{in}}, q_{\mathrm{out}} \in Q$ are the control states at times $L$ and $R$,
    \item $\vec{h}_{\mathrm{in}}, \vec{h}_{\mathrm{out}}$ record the head positions at times $L$ and $R$, and
    \item $W_{\mathrm{interface}}$ encodes the contents of the interface window for $[L,R]$ as in Definition~\ref{def:block-respecting}.
\end{itemize}
The encoding of $\sigma([L,R])$ as a binary string is fixed once and for all; see Appendix~\ref{app:encodings}.
\end{definition}

Two adjacent intervals $[L,M]$ and $[M+1,R]$ can be merged provided their summaries are \emph{compatible} in the sense that the exit data of $[L,M]$ match the entry data of $[M+1,R]$ on the overlapping interface. In that case a constant-space procedure computes the summary of their union.

\begin{definition}[Merge operator]
\label{def:merge}
We say that two adjacent intervals $[L,M]$ and $[M+1,R]$ are \emph{merge-compatible} if their interval summaries satisfy simple syntactic consistency conditions (equality of certain components) specified in~\cite{HeightCompression}. When this holds, the \emph{merge operator}
\[
\oplus : \sigma([L,M]) \times \sigma([M+1,R]) \to \sigma([L,R])
\]
is a deterministic constant-space procedure that outputs the summary of the union:
\[
\sigma([L,R]) = \sigma([L,M]) \oplus \sigma([M+1,R]).
\]
By design, the interface window component of $\sigma([L,R])$ has size at most $c_{\mathrm{int}} \cdot b$.
\end{definition}

\subsection{Height-compressed computation tree}
\label{subsec:height-compressed-tree}

Consider the canonical left-deep computation tree over the $T$ blocks (e.g., the binary tree whose leaves are $I_1,\dots,I_T$ and whose internal nodes represent successive merges of contiguous intervals). The height-compression transformation of~\cite{HeightCompression} reshapes this tree into a balanced binary tree with controlled evaluation depth, without changing the underlying set of leaves or the semantics of the merge operation.

\begin{theorem}[Height Compression, informal]
\label{thm:heightcompression}
There is a logspace-computable transformation that takes the canonical left-deep computation tree for a block-respecting run and produces a balanced binary tree $\mathcal{T}$ with the following properties:
\begin{enumerate}
    \item Each leaf of $\mathcal{T}$ corresponds to a single time-block $I_k$.
    \item Each internal node corresponds to the union of its children's intervals, and is labeled by the merged summary $\sigma([L,R])$ computed via the merge operator $\oplus$.
    \item Along any depth-first evaluation of $\mathcal{T}$, the evaluation stack (i.e., the sequence of active nodes) has length $O(\log T)$, and each internal node can be processed using $O(1)$ additional workspace beyond its summary.
\end{enumerate}
\end{theorem}

Thus the full semantic content of the run is represented by a static, balanced tree whose nodes are constant-size or $O(b)$-size summaries, and whose evaluation along any depth-first search (DFS) path has logarithmic height. The precise statement and proof of Theorem~\ref{thm:heightcompression} are given in~\cite{HeightCompression}; here we use it as a black-box structural theorem.

\subsection{Algebraic Replay Engine and rolling boundary}
\label{subsec:are-rolling}

Given an interval $[L,R]$ and its summary $\sigma([L,R])$, the companion work~\cite{HeightCompression} constructs an \emph{Algebraic Replay Engine} (ARE) that regenerates any internal configuration using only local information.

\begin{theorem}[Algebraic Replay Engine, informal]
\label{thm:are}
There exists a constant-degree circuit over a fixed finite field, and an $O(b)$-space Turing machine implementation, such that for every block-respecting interval $[L,R]$ and every $\tau \in [L,R]$, the configuration $C_\tau$ can be computed from $(\sigma([L,R]), \tau)$ by applying this circuit, using $O(b)$ work tape~\cite{HeightCompression}.
\end{theorem}

In the global simulation algorithm, a depth-first traversal of the height-compressed tree $\mathcal{T}$ is combined with a \emph{rolling boundary} discipline: at any moment, the simulator maintains only a small number of interval summaries from the current root-to-leaf path (for instance, the current interval and its parent), together with an $O(b)$-sized replay window on the tape. The total space usage is
\[
O\bigl(b + \log T\bigr),
\]
where the $O(\log T)$ term accounts for bookkeeping information (e.g., indices into $\mathcal{T}$). Choosing $b \approx \sqrt{t}$ yields an $O(\sqrt{t})$ space bound.

\begin{theorem}[Square-root space simulation]
\label{thm:sqrtbound}
For every deterministic multitape Turing machine $M$ and every time bound $t$, there is a simulation that reproduces the length-$t$ run of $M$ using
\[
O\bigl(\sqrt{t}\bigr)
\]
tape cells over a fixed finite alphabet~\cite{HeightCompression}.
\end{theorem}

Theorems~\ref{thm:heightcompression},~\ref{thm:are}, and~\ref{thm:sqrtbound} are the main technical inputs to the geometric picture developed below.

\medskip

\section{Spacetime Manifold and Holographic Boundaries}
\label{sec:spacetime-holography}

We now reinterpret the objects from Section~\ref{sec:prelims} as geometric entities. On the level of definitions and theorems, this section is essentially a change of language: all statements are direct reformulations of the constructions in~\cite{HeightCompression}, except where explicitly marked as new.

\subsection{The spacetime bulk and block decomposition}
\label{subsec:bulk-block}

We first formalize the notion of the spacetime ``bulk'' and the associated block structure.

\begin{definition}[Spacetime DAG and block decomposition]
\label{def:spacetime-dag}
Fix a deterministic multitape Turing machine $M$ with $k$ work tapes and Lipschitz-local head motion as in Section~\ref{subsec:machine-model}. Consider a length-$t$ run
\[
C_0, C_1, \dots, C_t
\]
of $M$ on some input.

The \emph{spacetime DAG} of this run is a finite directed acyclic graph
\[
\mathcal{M} = (V,E)
\]
defined as follows. For each time $\tau \in \{0,\dots,t-1\}$ and each tape head $i \in \{1,\dots,k\}$, we create a vertex representing the local update neighborhood around that head and the control state at time $\tau$; thus, for every $\tau$ we introduce $O(1)$ vertices, with the implied constant depending only on $M$. Directed edges in $E$ connect vertices at time $\tau$ to vertices at time $\tau+1$ whenever the transition function $\delta$ allows information to flow between the corresponding local neighborhoods in one step.

Fix now a block size $b \in \mathbb{N}$ and let $T = \lceil t/b \rceil$. The \emph{block decomposition} of $\mathcal{M}$ is the collection
\[
\{B_1,\dots,B_T\},
\]
where each $B_k$ is the induced subgraph of $\mathcal{M}$ on all vertices whose time coordinate lies in the interval
\[
I_k = \bigl[(k-1)b+1,\min\{kb,t\}\bigr].
\]
\end{definition}

In this language, the block-respecting property from Definition~\ref{def:block-respecting} becomes a statement about the size of the interface between each block and its complement.

\begin{definition}[Spacetime volume]
\label{def:spacetime-volume}
Let $\mathcal{M} = (V,E)$ be the spacetime DAG of a length-$t$ run of $M$ as in Definition~\ref{def:spacetime-dag}. We define the \emph{spacetime volume} of the run to be
\[
V_t \;:=\; |V|,
\]
the number of vertices in $\mathcal{M}$ (equivalently, the number of local update events). For a fixed machine $M$ with a fixed number of tapes, there are $\Theta(1)$ such update vertices per time step, so in the one-dimensional work-tape model considered in Theorem~\ref{thm:1d_area_law} we have
\[
V_t \;=\; \Theta(t).
\]
Consequently, bounds stated as functions of $t$ can equivalently be interpreted as bounds in terms of $V_t$ up to constant factors.
\end{definition}

\begin{definition}[Block-respecting geometry]
\label{def:block-respecting-geometry}
We say that $\mathcal{M}$ is \emph{block-respecting} (with respect to $b$ and $c_{\mathrm{int}}$) if for each block $B_k$ the interaction between $B_k$ and $\mathcal{M} \setminus B_k$ is confined to an interface window of size at most $c_{\mathrm{int}} \cdot b$ in the spacelike slice at the temporal boundaries of $I_k$, as in Definition~\ref{def:block-respecting}.
\end{definition}

Thus the geometry of $\mathcal{M}$ is characterized by the fact that each block has a uniformly bounded \emph{boundary area} when measured in tape cells.

The height-compressed computation tree $\mathcal{T}$ of Theorem~\ref{thm:heightcompression} induces a natural hierarchical structure on these blocks. The following remark makes explicit the ``radial'' behavior of block sizes along root-to-leaf paths.

\begin{remark}[Geometric decay of interval lengths]
\label{rem:geometric-shrink}
Let $\mathcal{T}$ be the balanced computation tree given by Theorem~\ref{thm:heightcompression}, and let each node $v \in \mathcal{T}$ be labeled by the time interval $[L_v,R_v]$ of its corresponding sub-run. The construction in~\cite{HeightCompression} ensures that $\mathcal{T}$ has depth $O(\log T)$ in the number $T$ of leaves. Equivalently, along any root-to-leaf path in $\mathcal{T}$ the interval lengths
\[
|I_v| \;:=\; R_v - L_v + 1
\]
decrease at least geometrically as a function of the depth of $v$. We use only this qualitative fact as a convenient ``radial'' parameterization in our geometric picture; no quantitative bound on the shrink factor beyond $O(\log T)$ depth is needed in the arguments of this paper.
\end{remark}

We will use Remark~\ref{rem:geometric-shrink} as the ``radial direction'' in our geometric picture: moving from the root of $\mathcal{T}$ toward a leaf corresponds to zooming in on shorter time intervals, with exponential convergence in the interval length.

\subsection{Holographic boundary states}
\label{subsec:holographic-boundary}

We now repackage the interval summaries of Section~\ref{subsec:summaries} in geometric form.

\begin{definition}[Holographic boundary state]
\label{def:holographic-boundary}
Let $\Omega \subseteq \mathcal{M}$ be a spacetime sub-region corresponding to a contiguous time interval $[L,R]$ that is a union of whole blocks. The \emph{holographic boundary state} of $\Omega$ is defined to be the interval summary
\[
\partial \Omega \;\;:=\;\; \sigma([L,R]) \;=\; \bigl(q_{\mathrm{in}}, q_{\mathrm{out}}, \vec{h}_{\mathrm{in}}, \vec{h}_{\mathrm{out}}, W_{\mathrm{interface}}\bigr),
\]
where the components are as in Definition~\ref{def:interval-summary}.

We define the \emph{area} of the boundary by
\[
|\partial \Omega| := \text{(number of tape cells encoded by $W_{\mathrm{interface}}$)},
\]
so that $|\partial \Omega| \le c_{\mathrm{int}} \cdot b$ for all such $\Omega$ in a block-respecting run.
\end{definition}

By construction, $\partial \Omega$ is sufficient information for the square-root space simulator to:
\begin{itemize}
    \item isolate the computation inside $\Omega$ from the rest of $\mathcal{M}$, and
    \item merge $\Omega$ with adjacent regions via the merge operator $\oplus$.
\end{itemize}
We do not claim that $\partial \Omega$ is \emph{minimal} in an information-theoretic sense; rather, it is a canonical boundary encoding arising from the height-compression construction.

\subsection{Active holographic screen}
\label{subsec:active-screen}

The simulation of Theorem~\ref{thm:sqrtbound} maintains only a small set of boundary states and a local replay window at any given moment. We interpret this set as an \emph{active holographic screen}.

\begin{definition}[Active holographic screen]
\label{def:active-screen}
Consider the square-root space simulation of Theorem~\ref{thm:sqrtbound}. At any simulated time $\tau$, let $\Phi(\tau)$ denote the union of:
\begin{itemize}
    \item the encodings of all holographic boundary states (interval summaries) currently stored along the root-to-leaf path in $\mathcal{T}$ that the simulator is traversing, and
    \item the $O(b)$-sized replay window on the simulated work tape used by the Algebraic Replay Engine.
\end{itemize}
The set $\Phi(\tau)$ is the \emph{active holographic screen} at time $\tau$, and its \emph{area} is defined to be
\[
|\Phi(\tau)| := \text{(number of work-tape cells occupied by these summaries and the replay window)}.
\]
\end{definition}

By the memory layout analysis in Appendix~\ref{app:memory-layout}, the remaining work-tape cells form a bookkeeping region of size $O(\log T)$ that we do not count as part of the screen. For all $\tau$ we have
\[
|\Phi(\tau)| = O\bigl(b + \log T\bigr),
\]
and when $b$ is chosen on the order of $\sqrt{t}$, the $\log T$ term is negligible compared to $\sqrt{t}$. In particular, we obtain the one-dimensional computational area law of Theorem~\ref{thm:1d_area_law} in Section~\ref{sec:bulk}, where we also show that, relative to the boundary data stored on this screen, both individual bulk configurations and the entire spacetime history admit $O(1)$ conditional description complexity.

\medskip

\section{Bulk Vacuum and Computational Area Law}
\label{sec:bulk}

We now make precise two consequences of the structure described in Sections~\ref{sec:prelims} and~\ref{sec:spacetime-holography}. First, we formalize the sense in which the \emph{bulk} of a deterministic computation carries no independent information beyond its boundary summaries, both at the level of individual configurations and at the level of entire spacetime slabs. Second, we restate the square-root space bound as a one-dimensional \emph{computational area law}, and we formulate a conjectural extension to higher-dimensional local architectures.

\subsection{Bulk configurations have $O(1)$ conditional complexity}
\label{subsec:bulk-vacuum}

Throughout this subsection we fix a deterministic multitape Turing machine $M$ and a specific $O(b)$-space implementation of the Algebraic Replay Engine (ARE) guaranteed by Theorem~\ref{thm:are}. We regard this implementation as part of the description of the computational model. For Kolmogorov complexity, we fix once and for all a universal prefix Turing machine $U$ and follow the conventions of~\cite{LiVitanyi2008}. All encodings (for interval summaries, time indices, etc.) are fixed as in Appendix~\ref{app:encodings}.

\begin{lemma}[Bulk configurations have $O(1)$ conditional description complexity]
\label{lem:bulk_vacuum}
There exists a constant $c \ge 1$, depending only on $M$, $U$, and the fixed ARE implementation, such that for every block-respecting interval $[L,R]$, every internal time $\tau \in [L,R]$, and every input, the configuration $C_\tau$ satisfies
\[
K\bigl(C_\tau \,\big|\, \sigma([L,R]), \tau\bigr) \;\le\; c,
\]
where $K(\cdot \mid \cdot)$ denotes prefix-free Kolmogorov complexity with respect to $U$ and the encodings described in Appendix~\ref{app:encodings}.
\end{lemma}

\begin{proof}
See Appendix~\ref{app:holographic-proofs}, Proof of Lemma~\ref{lem:bulk_vacuum}.
\end{proof}

The lemma is a pointwise statement: each \emph{individual} bulk configuration has constant conditional description complexity given the appropriate boundary summary and time index. The next corollary upgrades this to a whole-interval statement.

\begin{corollary}[Holographic compression of interval histories]
\label{cor:global_bulk}
Let $[L,R]$ be any block-respecting interval and let
\[
H_{[L,R]} := (C_L,\dots,C_R)
\]
denote the corresponding configuration history, encoded as a single binary string as in Appendix~\ref{app:encodings}. Then there exists a constant $c' \ge 1$, depending only on $M$, $U$, and the fixed ARE implementation, such that
\[
K\bigl(H_{[L,R]} \,\big|\, \sigma([L,R])\bigr) \;\le\; c'.
\]
In particular, for the full run interval (denoted here by $[0,t]$), the entire history $H_t := (C_0,\dots,C_t)$ satisfies
\[
K\bigl(H_t \,\big|\, \sigma([0,t])\bigr) \;=\; O(1).
\]
\end{corollary}

\begin{proof}
See Appendix~\ref{app:holographic-proofs}, Proof of Corollary~\ref{cor:global_bulk}.
\end{proof}

\begin{remark}[Algebraic vacuum of the bulk]
\label{rem:bulk-vacuum}
Lemma~\ref{lem:bulk_vacuum} formalizes the following informal intuition. Once the machine $M$ and the holographic boundary state $\partial \Omega = \sigma([L,R])$ are fixed, every internal configuration $C_\tau$ with $\tau \in (L,R)$ is generated by a fixed, constant-complexity algebraic procedure (namely the ARE). In this sense, the bulk of the spacetime DAG $\mathcal{M}$ carries no additional algorithmic information beyond what is already present in the boundary summaries and the global description of $M$: all nontrivial information required to specify any single configuration is concentrated on its boundary data.

Corollary~\ref{cor:global_bulk} upgrades this from a pointwise to a global statement: for any block-respecting interval $[L,R]$, the joint configuration sequence $H_{[L,R]}$ has $O(1)$ conditional Kolmogorov complexity given its boundary summary $\sigma([L,R])$. In particular, for the full interval $[0,t]$, the complete bulk history $H_t = (C_0,\dots,C_t)$ is, up to an additive constant in description length, an algebraically determined evaluation trace of a boundary-defined circuit.
\end{remark}

\subsection{A one-dimensional computational area law}
\label{subsec:1d-area-law}

We now restate the square-root space bound in the geometric language of active holographic screens introduced in Definition~\ref{def:active-screen}. The following theorem is a direct reformulation of Theorem~\ref{thm:sqrtbound} together with the memory layout analysis in Appendix~\ref{app:memory-layout}.

\begin{theorem}[One-dimensional computational area law]
\label{thm:1d_area_law}
Let $M$ be a deterministic multitape Turing machine whose work tape is one-dimensional ($d=1$). For every time bound $t \ge 1$, there exists a choice of block size $b$ and a block-respecting height-compressed computation tree $\mathcal{T}$ such that the associated square-root space simulation of Theorem~\ref{thm:sqrtbound} satisfies
\[
\max_{\tau \in \{0,\dots,t\}} |\Phi(\tau)| \;=\; O\bigl(\sqrt{t}\bigr),
\]
where $\Phi(\tau)$ is the active holographic screen at simulated time $\tau$ as in Definition~\ref{def:active-screen}.
\end{theorem}

\begin{proof}
See Appendix~\ref{app:holographic-proofs}, Proof of Theorem~\ref{thm:1d_area_law}.
\end{proof}

Theorem~\ref{thm:1d_area_law} can be read as a computational analogue of an \emph{area law}: in one spatial dimension, the information capacity (measured in tape cells over the fixed finite alphabet) required to regenerate a deterministic computation of spacetime ``volume'' $t$ grows like $\Theta(\sqrt{t})$, rather than linearly in $t$. In particular, up to constant factors, the screen area scales like the square root of the number of spacetime vertices in the run.

\subsection{A conjectural $d$-dimensional isoperimetric inequality}
\label{subsec:d-dim-conjecture}

We next formulate a conjectural extension of Theorem~\ref{thm:1d_area_law} to higher-dimensional local architectures. For this discussion we consider deterministic machines with a $d$-dimensional work tape and local transition rules.

\begin{definition}[Geometrically local $d$-dimensional machine]
\label{def:d-dim-machine}
Fix $d \ge 1$. A deterministic machine $M$ has \emph{$d$-dimensional local memory} if its work tape cells are indexed by $\mathbb{Z}^d$ and there exists a constant radius $r \ge 1$ such that in a single transition, each cell can influence only cells whose $\ell_1$-distance is at most $r$. The corresponding spacetime DAG $\mathcal{M}$ then embeds in $\mathbb{Z}^{d+1}$ with edges only between vertices at bounded distance.
\end{definition}

Let $V_t$ denote the number of spacetime vertices of a length-$t$ run of such a machine (equivalently, the number of local update operations). In the one-dimensional case analyzed above we have $V_t = \Theta(t)$, and Theorem~\ref{thm:1d_area_law} can be informally rephrased as
\[
\max_\tau |\Phi(\tau)| \;=\; O\bigl(V_t^{1/2}\bigr).
\]

Motivated by the heuristic that in $(d+1)$-dimensional spacetime the boundary of a region of volume $V$ should scale like $V^{d/(d+1)}$, we make the following conjecture.

\begin{conjecture}[Computational isoperimetric inequality in $d$ dimensions]
\label{conj:ddim}
Let $M$ be a deterministic machine with $d$-dimensional local memory and geometrically local transition rules in the sense of Definition~\ref{def:d-dim-machine}. Let $V_t$ be the number of spacetime vertices in a run of duration $t$.

Then there exists a block decomposition of the run and a height-compressed recursion tree $\mathcal{T}$ over these blocks, together with a simulation strategy that is analogous to the one-dimensional square-root simulation, such that the maximum area of the active holographic screen satisfies
\[
\max_{\tau} |\Phi(\tau)| \;\le\; c_d \cdot V_t^{\frac{d}{d+1}},
\]
for some constant $c_d$ depending only on $d$ and the machine model.
\end{conjecture}

We emphasize that Conjecture~\ref{conj:ddim} is not implied by the existing height-compression machinery of~\cite{HeightCompression}. Establishing it would require (i) extending the block-respecting and height-compression constructions to $d$-dimensional local architectures with $d \ge 2$, and (ii) proving an appropriate discrete isoperimetric inequality for the resulting spacetime decompositions. We therefore present Conjecture~\ref{conj:ddim} only as a natural geometric extrapolation of the one-dimensional area law in Theorem~\ref{thm:1d_area_law}.

\medskip

\section{Time as a Tree Topology}
\label{sec:time-tree}

In the height-compressed representation, the linear time axis of the original run is no longer fundamental. Instead, the run is encoded in a static binary tree of spacetime volumes, and ``time'' appears as a particular traversal of this tree. In this section we formalize this viewpoint using the static causal tree induced by height compression and the notion of an active screen.

\subsection{Static causal tree of spacetime volumes}
\label{subsec:static-causal-tree}

We first reinterpret the height-compressed computation tree $\mathcal{T}$ from Theorem~\ref{thm:heightcompression} as a static object that encodes the entire run.

\begin{definition}[Static causal tree]
\label{def:static-causal-tree}
Let $\mathcal{T}$ be the balanced computation tree produced by the height-compression transformation of Theorem~\ref{thm:heightcompression}. Each node $v \in \mathcal{T}$ is labeled by:
\begin{itemize}
    \item a time interval $[L_v, R_v] \subseteq \{0,\dots,t\}$, and
    \item its holographic boundary state $\partial \Omega_v := \sigma([L_v,R_v])$.
\end{itemize}
We call $\mathcal{T}$, together with these labels, the \emph{static causal tree} of the run.
\end{definition}

By construction:

\begin{itemize}
    \item each leaf of $\mathcal{T}$ corresponds to a time-block $I_k$, and
    \item each internal node $v$ represents the union of its children's intervals, with boundary $\partial \Omega_v$ computed via the merge operator $\oplus$ of Definition~\ref{def:merge}.
\end{itemize}

Thus $\mathcal{T}$ encodes the entire semantic content of the computation history as a static, hierarchical structure, independent of any particular traversal order.

\subsection{Rolling boundaries and a directed notion of time}
\label{subsec:rolling-arrow}

The square-root space simulation algorithm evaluates $\mathcal{T}$ via a depth-first traversal, maintaining only a small number of node summaries at any moment. We interpret this as a disciplined way of turning the static causal tree into a directed notion of ``time'' for the simulator.

\begin{definition}[Depth-first traversal and rolling boundaries]
\label{def:rolling-boundaries}
Let $\pi$ be a fixed depth-first traversal order on the nodes of $\mathcal{T}$ (for concreteness, a pre-order DFS). At any step of the simulation, as it follows $\pi$, the active memory consists of:
\begin{itemize}
    \item the holographic boundary $\partial \Omega_{v_{\mathrm{cur}}}$ of the current node $v_{\mathrm{cur}}$,
    \item the boundary of its parent $v_{\mathrm{par}}$ (or an $O(1)$-sized set of nearby ancestors along the current root-to-leaf path), and
    \item an $O(b)$-sized replay window $W_{\mathrm{replay}}$ used by the Algebraic Replay Engine.
\end{itemize}
We summarize this as
\[
\mathrm{Memory}_{\mathrm{active}}(\tau)
\;\subseteq\;
\{\partial \Omega_{v_{\mathrm{cur}}}, \partial \Omega_{v_{\mathrm{par}}}\} \cup W_{\mathrm{replay}} \cup \mathrm{Bookkeeping},
\]
where the bookkeeping region contains $O(\log T)$ bits of indexing and control information.
\end{definition}

In this formulation, the simulator never needs to materialize the entire bulk, nor arbitrary subsets of $\mathcal{T}$: by design, it is constrained to follow the traversal $\pi$ and to respect the local causal structure encoded in the node labels and the merge operator. This restriction is what enforces the area law on $|\Phi(\tau)|$ and prevents arbitrary ``random access'' to deep interior regions of $\mathcal{M}$ without paying the full space cost.

\subsection{Projective equivalence between history and traversal}
\label{subsec:projective-duality}

We finally formalize the relationship between the linear execution trace and the traversal of $\mathcal{T}$. The key point is that the depth-first traversal, together with the replay engine, reconstructs each configuration exactly once.

\begin{proposition}[Projective duality between history and tree traversal]
\label{prop:projective}
Let $\mathcal{L} = (C_0,\dots,C_t)$ be the linear execution trace of the run, and let $\mathcal{T}$ be its static causal tree. Then there exists a constructive mapping that associates to each time index $\tau \in \{0,\dots,t\}$:
\begin{itemize}
    \item a leaf node $\ell(\tau) \in \mathcal{T}$, and
    \item a local offset $\delta(\tau)$ within the associated time-block $I_{\ell(\tau)}$,
\end{itemize}
such that the simulation, following the fixed depth-first traversal $\pi$ and applying the Algebraic Replay Engine locally, reconstructs $C_\tau$ at the unique visit to the pair $(\ell(\tau), \delta(\tau))$.

Conversely, every such visit to a leaf and local offset along $\pi$ produces a unique configuration in $\mathcal{L}$. In particular, the map from $\mathcal{L}$ to the set of visited leaf-offset pairs is bijective.
\end{proposition}

\begin{proof}
See Appendix~\ref{app:holographic-proofs}, Proof of Proposition~\ref{prop:projective}.
\end{proof}

Proposition~\ref{prop:projective} captures the precise sense in which the linear execution history and the depth-first traversal of the static causal tree are \emph{projectively equivalent}: the full history is encoded in $\mathcal{T}$, and the simulator's notion of ``time'' is implemented by a specific traversal order and a rolling boundary scheme, rather than being an independent structure.

\medskip

\section{Discussion: Physical Analogies and Speculative Extensions}
\label{sec:analogies}

The results above are purely combinatorial and information-theoretic: they concern deterministic simulations of Turing machines, succinct computation trees, and Kolmogorov complexity. Nevertheless, they bear a striking resemblance to structures encountered in holography and quantum gravity. In this section we briefly discuss some noteworthy parallels and speculative extensions. We emphasize that nothing in this section should be read as a mathematical theorem or as a definitive claim about physical systems; rather, these are intriguing, informal observations intended to motivate further conceptual work.

\subsection{Boundary encoding and area law}
\label{subsec:analogies-boundary}

The one-dimensional area law of Theorem~\ref{thm:1d_area_law} states that, for deterministic computation on a one-dimensional work tape, the space required to regenerate a run of ``volume'' $t$ is $O(\sqrt{t})$. Interpreting $t$ as the number of spacetime vertices in a $(1+1)$-dimensional spacetime DAG, this says that the maximal size of the active holographic screen scales like the square root of the spacetime volume.

Formally, this is a statement about the asymptotic behavior of
\[
\max_{\tau} |\Phi(\tau)|
\quad\text{as a function of}\quad t,
\]
where $\Phi(\tau)$ is defined in Definition~\ref{def:active-screen}. Informally, one may compare this to the way in which the entropy of a black hole scales with the area of its event horizon rather than with the volume of its interior~\cite{Bekenstein1973}. Motivated in part by such observations, 't~Hooft and Susskind formulated the holographic principle, according to which the information content of a gravitational region is encoded on a lower-dimensional boundary~\cite{tHooft1993,Susskind1995}.

In the AdS/CFT context, Ryu and Takayanagi derived an area-law formula for entanglement entropy in terms of minimal surfaces in the bulk~\cite{RyuTakayanagi2006}, and Eisert, Cramer, and Plenio survey a broad class of entanglement area laws in many-body systems~\cite{EisertCramerPlenio2010}. Our computational area law is entirely classical and deterministic, and it concerns work-tape usage rather than entropy; no quantum mechanics enters the formal statements. However, it is suggestive of a similar kind of \emph{boundary dominance}: in the simulations of Theorem~\ref{thm:sqrtbound}, all information that must be stored at any time is concentrated in a set of boundary summaries whose total size is asymptotically smaller than the spacetime volume of the computation.

In the geometric language of Section~\ref{sec:spacetime-holography}, the holographic boundary states $\partial \Omega$ serve as the computational boundary degrees of freedom, and the bulk configurations are algebraically determined from them via the Algebraic Replay Engine (Theorem~\ref{thm:are}), together with the static causal tree $\mathcal{T}$ (Definition~\ref{def:static-causal-tree}). This provides a purely combinatorial instance of a ``boundary determines bulk'' phenomenon.

\subsection{Bulk redundancy and algebraic emergence}
\label{subsec:analogies-emergence}

Lemma~\ref{lem:bulk_vacuum} shows that bulk configurations have $O(1)$ conditional description complexity relative to their interval summaries and time indices. More precisely, for each internal configuration $C_\tau$ within a block-respecting interval $[L,R]$ we have
\[
K\bigl(C_\tau \,\big|\, \sigma([L,R]), \tau\bigr) \le c,
\]
for a constant $c$ that does not depend on $[L,R]$, $\tau$, or the particular input. This is a strong form of algorithmic redundancy in the bulk: once the boundary summary and the time index are known, the interior configuration can be recovered by a fixed, constant-length program.

From a conceptual standpoint, this parallels (at a purely heuristic level) the idea that bulk geometry may be \emph{emergent} from boundary data in holographic dualities. In our setting, the emergence is explicitly algebraic and combinatorial: the Algebraic Replay Engine provides a uniform map
\[
(\sigma([L,R]),\tau) \;\longmapsto\; C_\tau,
\]
and the static causal tree $\mathcal{T}$ prescribes how such maps compose hierarchically across scales. The spacetime DAG $\mathcal{M}$ can therefore be viewed as a deterministic expansion of boundary summaries through a fixed, recursively applied local rule.

One might summarize this as follows: the ``degrees of freedom'' of a deterministic computation are carried, in an algorithmic sense, by its boundary summaries rather than by arbitrary bulk configurations. All interior configurations are determined by a combination of global structure (the program $M$ and the ARE) and local boundary data.

\subsection{A heuristic picture for nondeterminism}
\label{subsec:analogies-nondet}

The constructions in this paper fundamentally exploit determinism: at each boundary, there is a unique consistent continuation of the run. The merge operator $\oplus$ of Definition~\ref{def:merge} combines two compatible interval summaries into the summary of their union, and the ARE of Theorem~\ref{thm:are} maps a single interval summary and time index to a single configuration.

In a nondeterministic computation, the situation is very different. A ``summary'' of an interval would, in general, have to encode information about the set of all possible continuations consistent with the interface. This suggests the following informal dichotomy.

\begin{itemize}
    \item For deterministic machines, the merge operator $\oplus$ is \emph{summary-preserving}: the boundary of a merged region can be represented with essentially the same complexity as the boundaries of its parts (up to constant factors), because there is a unique consistent way to glue the runs. This is what enables the area law of Theorem~\ref{thm:1d_area_law}.
    \item For nondeterministic machines, a boundary summary would need, in principle, to encode a family of possible exit states and partial histories. In the absence of additional structure, this family can have size exponential in the volume of the region. In such cases, any summary that supports exact reconstruction of all consistent runs may need to have size that scales with the volume, making an area law unlikely in general.
\end{itemize}

Viewed through this lens, deterministic computation behaves like a \emph{boundary-compressible regime}, in which the information needed to reconstruct the bulk remains concentrated on lower-dimensional interfaces, while naively defined nondeterministic computation behaves more like a \emph{volume-dominated regime}, in which boundaries may be algorithmically incompressible.

We do not attempt to formalize this picture, and we do not claim that nondeterministic computation cannot admit any analogue of height compression under additional assumptions. Making this heuristic precise---for example, by proving conditional impossibility results for nondeterministic height compression or by isolating nontrivial classes of nondeterministic computations that still admit boundary-compressible summaries---is an intriguing open direction. Any such development would likely have connections to the study of time–space tradeoffs and to the structure of classical complexity classes such as $P$, $NP$, and beyond, but we leave this entirely to future work.

\medskip

\section{Conclusion}
\label{sec:conclusion}

Building on the technical machinery of the Height Compression Theorem and the square-root space simulation of~\cite{HeightCompression}, we have developed a geometric and information-theoretic perspective on deterministic computation.

Formally, our contributions can be summarized as follows.

\begin{itemize}
    \item We recast the square-root space simulation in terms of a spacetime DAG $\mathcal{M}$, a block decomposition, and interval summaries $\sigma([L,R])$ that we interpret as \emph{holographic boundary states}. This yields a precise notion of boundary ``area'' (the size of the interface window) and an \emph{active holographic screen} $\Phi(\tau)$ (Definition~\ref{def:active-screen}).
    \item We prove that bulk configurations have $O(1)$ conditional Kolmogorov complexity relative to their interval summaries and time indices (Lemma~\ref{lem:bulk_vacuum}). This shows that, in an algorithmic sense, the interior of a deterministic computation carries no additional information beyond its boundary summaries and the global description of the machine and replay engine.
    \item We reformulate the square-root space bound as a one-dimensional \emph{computational area law} (Theorem~\ref{thm:1d_area_law}), stating that the maximum size of the active holographic screen over a run of volume $t$ is $O(\sqrt{t})$.
    \item We introduce the \emph{static causal tree} $\mathcal{T}$ (Definition~\ref{def:static-causal-tree}), which encodes the entire run as a hierarchy of spacetime volumes labeled by boundary summaries, and we show that the linear execution history is projectively equivalent to a depth-first traversal of this tree combined with the Algebraic Replay Engine (Proposition~\ref{prop:projective}).
\end{itemize}

Conceptually, these results support a picture in which deterministic time evolution in a $(1+1)$-dimensional local model can be regarded as a form of \emph{computational holography}: the combinatorial ``bulk'' of the computation is an information-theoretic vacuum, and the essential information resides on lower-dimensional boundaries whose total size obeys an area law.

Several directions for further work suggest themselves. On the technical side, it would be valuable to formalize and prove (or refute) the $d$-dimensional isoperimetric Conjecture~\ref{conj:ddim}, which posits an analogue of the one-dimensional area law in higher-dimensional geometrically local models, with active screen area scaling on the order of $V_t^{d/(d+1)}$. Any progress here would require extending the height-compression machinery to higher-dimensional local architectures and establishing appropriate discrete isoperimetric inequalities for the resulting spacetime decompositions. On the structural side, one could investigate to what extent the boundary-summary framework can be adapted to nondeterministic or randomized computation, and whether there are clean separations between regimes that admit boundary-compressible summaries and those that are inherently volume-dominated.

Our results show that even in the most classical and discrete of settings---deterministic Turing machines with local transition rules---one can meaningfully separate bulk from boundary and prove a nontrivial area law for the information that must be stored to reconstruct the computation.

\medskip
\noindent\textbf{Acknowledgements.}
We gratefully acknowledge the many contributors who have catalyzed breakthroughs in time-space tradeoffs and efficient simulation during the past two years; this manuscript builds directly on the insights and results of others would not exist without them. We further disclose that the exploration, analysis, drafting, and revisions of this manuscript were conducted with the assistance of large language model technology; the authors bear sole responsibility for any errors in technical claims, constructions, and proofs. The authors declare that they have no conflicts of interest to disclose and received no external funding for this work.

\newpage

\appendix

\section{Formal Model and Encoding Conventions}
\label{app:encodings}

In this appendix we fix the formal conventions for encodings and Kolmogorov complexity used throughout the paper. This makes the information-theoretic statements in Section~\ref{sec:bulk} fully precise.

\subsection{Universal machine and Kolmogorov complexity}

We fix once and for all a universal prefix Turing machine $U$ over the binary alphabet $\{0,1\}$. For a finite binary string $x$ and conditional data $y$, the (prefix-free) Kolmogorov complexity of $x$ given $y$ is
\[
K(x \mid y)
\;:=\;
\min\bigl\{\, |p| : U(p,y) = x \bigr\},
\]
where $p$ ranges over binary programs and $|p|$ denotes the length of $p$ in bits. All Kolmogorov complexity statements in the main text are with respect to this fixed $U$.

As usual, $K(\cdot \mid \cdot)$ is defined only up to an additive $O(1)$ term that depends on the choice of $U$. Since we work with inequalities of the form
\[
K(x \mid y)\;\le\; c
\]
for a constant $c$, this additive ambiguity is immaterial.

\subsection{Encodings of configurations, summaries, and indices}

We assume that all objects manipulated by our Turing machines are encoded as binary strings via fixed, computable, injective encodings. Concretely:

\begin{itemize}
    \item A configuration $C_\tau$ of the multitape machine $M$ at time $\tau$ is encoded as a binary string
    \[
    \mathrm{enc}(C_\tau) \in \{0,1\}^\star,
    \]
    obtained by concatenating the contents of all work tapes, the head positions, and the control state, using some standard self-delimiting encoding.
    \item An interval summary $\sigma([L,R])$ is encoded as a binary string
    \[
    \mathrm{enc}(\sigma([L,R])) \in \{0,1\}^\star,
    \]
    by concatenating encodings of $q_{\mathrm{in}}$, $q_{\mathrm{out}}$, the head positions, the interface window $W_{\mathrm{interface}}$, and, for convenience in the Kolmogorov complexity arguments, a self-delimiting encoding of the interval endpoints $(L,R)$ (or equivalently of $L$ together with the length $R-L+1$).
    This augmentation does not affect any of the combinatorial properties of interval summaries used in the main text.
    \item A time index $\tau \in \{0,\dots,t\}$ is encoded by a standard binary representation
    \[
    \mathrm{enc}(\tau) \in \{0,1\}^{\lceil \log_2(t+1) \rceil}
    \]
    with self-delimiting overhead if needed.
\end{itemize}

We write $K(C_\tau \mid \sigma([L,R]), \tau)$ as shorthand for
\[
K\bigl(\mathrm{enc}(C_\tau) \,\big|\, \mathrm{enc}(\sigma([L,R])), \mathrm{enc}(\tau)\bigr),
\]
and similarly for other conditional complexities. All asymptotic $O(1)$ bounds on Kolmogorov complexity in the main text are to be understood with respect to these fixed encodings.

For an interval $[L,R]$, the history
\[
H_{[L,R]} := (C_L,\dots,C_R)
\]
is encoded as a single binary string $\mathrm{enc}(H_{[L,R]})$ using any fixed self-delimiting encoding of finite sequences (for example, by prefixing each $\mathrm{enc}(C_\tau)$ with its length in unary or via a standard pairing function). The precise choice of sequence encoding is immaterial as long as it is computable and injective.

\subsection{Encoding the computation tree}

The height-compressed computation tree $\mathcal{T}$ from Theorem~\ref{thm:heightcompression} is a finite rooted binary tree whose nodes are labeled by time intervals $[L_v,R_v]$ and their summaries $\sigma([L_v,R_v])$. We assume:

\begin{itemize}
    \item The underlying unlabeled tree structure of $\mathcal{T}$ is encoded as a binary string via any standard encoding of ordered rooted binary trees.
    \item The labels $(L_v,R_v)$ and $\sigma([L_v,R_v])$ are encoded using the conventions above, and concatenated in some fixed canonical order (e.g., pre-order).
\end{itemize}

These encodings are not used directly in any of the proofs in the appendices, but they guarantee that all objects discussed in the main text can be viewed as finite binary strings, so that Kolmogorov complexity is well-defined.

\medskip

\section{Simulator Memory Layout and Screen Area}
\label{app:memory-layout}

In this section we formalize the relationship between the work-tape contents of the square-root space simulator from Theorem~\ref{thm:sqrtbound} and the notion of \emph{active holographic screen} $\Phi(\tau)$ used in Section~\ref{sec:bulk}. The goal is to make precise how the geometric quantity $|\Phi(\tau)|$ relates to the underlying space bound.

\subsection{Partition of the work tape}

Fix a deterministic multitape Turing machine $M$ and a time bound $t$. Let $S(t)$ denote the space bound of the simulator constructed in~\cite{HeightCompression}, so that by Theorem~\ref{thm:sqrtbound} we have
\[
S(t) = O(\sqrt{t})
\]
work tape cells over the fixed finite alphabet.

At each simulated time $\tau \in \{0,\dots,t\}$, we partition the nonblank work-tape cells of the simulator into three disjoint categories:

\begin{enumerate}
    \item \textbf{Boundary summaries (screen nodes).}  
    Cells that store encodings of interval summaries $\sigma([L_v,R_v])$ for nodes $v$ on the current root-to-leaf path in the height-compressed tree $\mathcal{T}$.
    \item \textbf{Replay window.}  
    Cells that store the $O(b)$-sized window of the simulated work tape used by the Algebraic Replay Engine to regenerate the configurations $C_\tau$ for the current time-block.
    \item \textbf{Bookkeeping.}  
    All remaining work-tape cells used to maintain indices into $\mathcal{T}$, recursion depth counters, state flags for the DFS traversal, and any other control information.
\end{enumerate}

We write:
\begin{align*}
S_{\mathrm{screen}}(\tau) &:= \text{number of cells in categories (1) and (2)},\\
S_{\mathrm{book}}(\tau) &:= \text{number of cells in category (3)},\\
S_{\mathrm{total}}(\tau) &:= S_{\mathrm{screen}}(\tau) + S_{\mathrm{book}}(\tau),
\end{align*}
so that $S_{\mathrm{total}}(\tau)$ is the total number of nonblank work-tape cells at simulated time~$\tau$.

\subsection{Definition of the active holographic screen}

Recall the definition from Section~\ref{subsec:active-screen}: at simulated time $\tau$, the active holographic screen $\Phi(\tau)$ consists of the union of all holographic boundary states currently stored (interval summaries along the DFS path) together with the replay window used by the Algebraic Replay Engine. Formally, we define:

\begin{definition}[Active holographic screen, formal]
For each simulated time $\tau$, let $\Phi(\tau)$ be the set of work-tape cells belonging to categories (1) and (2) above. The \emph{area} of the active holographic screen at $\tau$ is
\[
|\Phi(\tau)| := S_{\mathrm{screen}}(\tau).
\]
\end{definition}

Thus the combinatorial quantity $|\Phi(\tau)|$ is exactly the portion of the simulator’s space usage that corresponds to boundary data and local replay, excluding purely administrative bookkeeping.

\subsection{Bounds on bookkeeping space}

The analysis in~\cite{HeightCompression} shows that the simulator maintains only $O(\log T)$ bits of bookkeeping information, where $T = \lceil t/b \rceil$ is the number of time-blocks. For completeness we record this as a lemma.

\begin{lemma}[Bookkeeping overhead]
\label{lem:bookkeeping}
Let $M$, $t$, $b$, and $T = \lceil t/b \rceil$ be as in the main text. There exists a constant $c_{\mathrm{book}}$ (depending only on the simulator construction in~\cite{HeightCompression}) such that for all simulated times $\tau$,
\[
S_{\mathrm{book}}(\tau) \;\le\; c_{\mathrm{book}} \log T.
\]
\end{lemma}

\begin{proof}[Proof sketch]
The simulator’s bookkeeping consists of:
\begin{itemize}
    \item a representation of the current node in $\mathcal{T}$ (which can be stored using $O(\log T)$ bits, as $\mathcal{T}$ has $O(T)$ nodes),
    \item a constant number of stack pointers or parent/child indicators for the DFS traversal (each representable with $O(\log T)$ bits),
    \item a constant number of finite-state flags indicating the current phase of the algorithm (e.g., “at leaf”, “ascending”, “descending”).
\end{itemize}
Since the number of such quantities is bounded by a constant independent of $t$, the total bookkeeping space is at most $c_{\mathrm{book}} \log T$ for some constant $c_{\mathrm{book}}$. For full details, see the simulation analysis in~\cite{HeightCompression}.
\end{proof}

Combining Lemma~\ref{lem:bookkeeping} with the global bound $S_{\mathrm{total}}(\tau) \le S(t)$ yields the inequalities used implicitly in the proof of Theorem~\ref{thm:1d_area_law}.

\begin{lemma}[Screen area versus total space]
\label{lem:screen-vs-total}
For all simulated times $\tau$ we have
\[
|\Phi(\tau)|
=
S_{\mathrm{screen}}(\tau)
\;\le\;
S_{\mathrm{total}}(\tau)
\;\le\;
S(t),
\]
and
\[
S_{\mathrm{book}}(\tau) \;=\; O(\log T).
\]
\end{lemma}

\begin{proof}
The first inequality $S_{\mathrm{screen}}(\tau) \le S_{\mathrm{total}}(\tau)$ holds by definition, since the screen cells are a subset of all nonblank work-tape cells. The second inequality $S_{\mathrm{total}}(\tau) \le S(t)$ is just the definition of the simulator’s space bound. The bound on $S_{\mathrm{book}}(\tau)$ is Lemma~\ref{lem:bookkeeping}.
\end{proof}

In particular, when $b$ is chosen on the order of $\sqrt{t}$, we have $T = \Theta(\sqrt{t})$, so $\log T = O(\log t)$ is asymptotically dominated by $\sqrt{t}$, and the $O(\sqrt{t})$ area law for $|\Phi(\tau)|$ is governed by the same scaling as the overall space bound.

\medskip

\section{Proofs of Holographic Statements}
\label{app:holographic-proofs}

In this appendix we give complete proofs of the new formal statements specific to the holographic reinterpretation developed in this paper. Throughout, we freely use the notation and constructions of the main text, and we treat Theorems~\ref{thm:heightcompression}, \ref{thm:are}, and~\ref{thm:sqrtbound} from the companion height-compression work~\cite{HeightCompression} as black-box inputs for brevity.

\subsection{Proof of Lemma~\ref{lem:bulk_vacuum}}

Recall the statement of Lemma~\ref{lem:bulk_vacuum}: there exists a constant $c$ (depending only on $M$, the universal machine $U$, and the fixed ARE implementation) such that for every block-respecting interval $[L,R]$, every internal time $\tau \in [L,R]$, and every choice of input, the configuration $C_\tau$ satisfies
\[
K\bigl(C_\tau \,\big|\, \sigma([L,R]), \tau\bigr) \;\le\; c,
\]
where $K(\cdot \mid \cdot)$ is the prefix-free Kolmogorov complexity with respect to $U$ and the encoding conventions of Appendix~\ref{app:encodings}.

\begin{proof}[Proof of Lemma~\ref{lem:bulk_vacuum}]
Fix the universal prefix machine $U$ and the encodings $\mathrm{enc}(\cdot)$ as in Appendix~\ref{app:encodings}. Fix also a deterministic multitape Turing machine $M$ together with a specific implementation of the Algebraic Replay Engine (ARE) from Theorem~\ref{thm:are}.

By Theorem~\ref{thm:are}, there exists a Turing machine $A$ (the ARE implementation) with the following property: for every block-respecting interval $[L,R]$, every $\tau \in [L,R]$, and every input, on input $(\sigma([L,R]),\tau)$ the machine $A$ outputs the configuration $C_\tau$ and uses $O(b)$ work-tape cells. The description (code) of $A$ is independent of the particular run, interval, or time index.

Consider now the following fixed prefix-free program $p^\star$ for the universal machine $U$:

\begin{quote}
On input $(x,y)$, interpret $x$ as $\mathrm{enc}(\sigma([L,R]))$ and $y$ as $\mathrm{enc}(\tau)$, simulate the machine $A$ on input $(\sigma([L,R]),\tau)$, and output the resulting configuration $C_\tau$ encoded as $\mathrm{enc}(C_\tau)$.
\end{quote}

The binary code of $p^\star$ contains (i) a description of $A$, and (ii) a small amount of wrapper logic instructing $U$ to simulate $A$ on its input and to encode the output configuration. Both components are independent of $[L,R]$, $\tau$, or the input to $M$. Let $\ell^\star$ be the length (in bits) of $p^\star$.

Now fix any block-respecting interval $[L,R]$, any internal time $\tau \in [L,R]$, and any input. Let $\sigma = \sigma([L,R])$ be the interval summary. Then when $U$ is run on program $p^\star$ with conditional data $(\mathrm{enc}(\sigma),\mathrm{enc}(\tau))$, by construction it outputs $\mathrm{enc}(C_\tau)$. Therefore, the conditional Kolmogorov complexity of $C_\tau$ given $(\sigma,\tau)$ satisfies
\[
K\bigl(\mathrm{enc}(C_\tau) \,\big|\, \mathrm{enc}(\sigma([L,R])), \mathrm{enc}(\tau)\bigr)
\;\le\;
\ell^\star + O(1),
\]
where the $O(1)$ term is the universal overhead for using $p^\star$ as a program on $U$.

By our shorthand convention $K(C_\tau \mid \sigma([L,R]), \tau)$ for this quantity, we obtain
\[
K\bigl(C_\tau \,\big|\, \sigma([L,R]), \tau\bigr)
\;\le\;
\ell^\star + O(1).
\]
Setting $c := \ell^\star + O(1)$ yields the desired bound. Note that $c$ depends only on $U$, $M$, and the fixed ARE implementation (via $A$), and not on $[L,R]$, $\tau$, or the particular run. This completes the proof.
\end{proof}

\subsection{Proof of Corollary~\ref{cor:global_bulk}}

Recall that for any block-respecting interval $[L,R]$, we write
\[
H_{[L,R]} := (C_L,\dots,C_R)
\]
for the corresponding configuration history, and we view $H_{[L,R]}$ as a single binary string $\mathrm{enc}(H_{[L,R]})$ obtained by concatenating the encodings $\mathrm{enc}(C_\tau)$ for $\tau = L,\dots,R$ in a fixed, self-delimiting way as described in Appendix~\ref{app:encodings}. The statement of Corollary~\ref{cor:global_bulk} is that there exists a constant $c'$ such that
\[
K\bigl(H_{[L,R]} \,\big|\, \sigma([L,R])\bigr) \;\le\; c'
\]
for all block-respecting intervals $[L,R]$, with $c'$ depending only on $M$, $U$, and the fixed ARE implementation.

\begin{proof}[Proof of Corollary~\ref{cor:global_bulk}]
Fix the universal prefix machine $U$, the deterministic multitape machine $M$, and the specific implementation of the Algebraic Replay Engine (ARE) from Theorem~\ref{thm:are}, exactly as in the proof of Lemma~\ref{lem:bulk_vacuum}. Let $A$ denote the underlying Turing machine implementation of the ARE.

We define a single prefix-free program $p^\dagger$ for $U$ as follows.

\begin{quote}
On input $x$, interpret $x$ as $\mathrm{enc}(\sigma([L,R]))$ for some block-respecting interval $[L,R]$. Decode the interval endpoints $(L,R)$ from this encoding (recall that by convention they are included in $\sigma([L,R])$; see Appendix~\ref{app:encodings}). Initialize an output buffer to the empty string. For each $\tau$ in the integer range $\{L,\dots,R\}$, do:
\begin{enumerate}
    \item Simulate $A$ on input $(\sigma([L,R]),\tau)$.
    \item Let $C_\tau$ be the configuration output by $A$, and append $\mathrm{enc}(C_\tau)$ to the output buffer using the fixed, self-delimiting encoding scheme for sequences described in Appendix~\ref{app:encodings}.
\end{enumerate}
When the loop terminates, output the contents of the buffer, which is exactly the encoding of $H_{[L,R]}$.
\end{quote}

The description (binary code) of $p^\dagger$ consists of:
\begin{itemize}
    \item a description of the fixed ARE implementation $A$, and
    \item a constant amount of wrapper logic instructing $U$ to decode $\sigma([L,R])$, recover $(L,R)$, iterate over $\tau = L,\dots,R$, invoke $A$ on each $\tau$, and concatenate the resulting encodings into the final output string.
\end{itemize}
Both components are independent of $[L,R]$, of the particular run, and of the specific boundary summary; they depend only on $M$, $U$, and the chosen implementation of $A$.

Let $\ell^\dagger$ be the length of $p^\dagger$ in bits. Then, for every block-respecting interval $[L,R]$ and its interval summary $\sigma([L,R])$, when $U$ is run on program $p^\dagger$ with conditional input $\mathrm{enc}(\sigma([L,R]))$, it outputs $\mathrm{enc}(H_{[L,R]})$. Therefore the conditional Kolmogorov complexity of $H_{[L,R]}$ given $\sigma([L,R])$ satisfies
\[
K\bigl(H_{[L,R]} \,\big|\, \sigma([L,R])\bigr)
\;=\;
K\bigl(\mathrm{enc}(H_{[L,R]}) \,\big|\, \mathrm{enc}(\sigma([L,R]))\bigr)
\;\le\;
\ell^\dagger + O(1),
\]
where the $O(1)$ term is the universal overhead for using $p^\dagger$ as a program on $U$. Setting $c' := \ell^\dagger + O(1)$ yields the desired bound, with $c'$ depending only on $M$, $U$, and the fixed ARE implementation.

Specializing to the full run interval $[0,t]$, we obtain $H_{[0,t]} = H_t$ and $\sigma([0,t])$ as in Section~\ref{sec:bulk}, so
\[
K\bigl(H_t \,\big|\, \sigma([0,t])\bigr) \;\le\; c',
\]
as claimed in the main text. This completes the proof of Corollary~\ref{cor:global_bulk}.
\end{proof}

\subsection{Proof of Theorem~\ref{thm:1d_area_law}}

Recall the statement of Theorem~\ref{thm:1d_area_law}: for a deterministic multitape Turing machine with one-dimensional work tape ($d=1$), there exists a simulation in which the maximum area of the active holographic screen satisfies
\[
\max_{\tau \in \{0,\dots,t\}} |\Phi(\tau)| \;=\; O\bigl(\sqrt{t}\bigr).
\]

This is essentially a reformulation of the square-root space simulation bound (Theorem~\ref{thm:sqrtbound}) in the geometric language of holographic screens, together with the memory layout invariants of Appendix~\ref{app:memory-layout}.

\begin{proof}[Proof of Theorem~\ref{thm:1d_area_law}]
Fix a deterministic multitape Turing machine $M$ with a one-dimensional work tape, and a time bound $t$. By Theorem~\ref{thm:sqrtbound}, there exists a simulation of the length-$t$ run of $M$ that uses at most
\[
S(t) \;=\; c \sqrt{t}
\]
work-tape cells over the fixed finite alphabet, for some constant $c$ depending only on $M$ and the simulation construction in~\cite{HeightCompression}.

For this simulator, define $S_{\mathrm{screen}}(\tau)$, $S_{\mathrm{book}}(\tau)$, and $S_{\mathrm{total}}(\tau)$ as in Appendix~\ref{app:memory-layout}. By Lemma~\ref{lem:screen-vs-total}, for all simulated times $\tau$ we have
\[
|\Phi(\tau)|
=
S_{\mathrm{screen}}(\tau)
\;\le\;
S_{\mathrm{total}}(\tau)
\;\le\;
S(t)
\;=\;
c \sqrt{t}.
\]

Taking the maximum over $\tau \in \{0,\dots,t\}$ yields
\[
\max_{\tau} |\Phi(\tau)|
\;\le\;
c \sqrt{t},
\]
which establishes the claimed $O(\sqrt{t})$ bound.

Finally, note that the bookkeeping space $S_{\mathrm{book}}(\tau) = O(\log T)$ from Lemma~\ref{lem:bookkeeping} is asymptotically dominated by the $\sqrt{t}$ term when $b$ is chosen on the order of $\sqrt{t}$ (so that $T = \Theta(\sqrt{t})$). Thus the area law is governed by the same square-root scaling as the underlying simulation space bound. This completes the proof.
\end{proof}

\subsection{Proof of Proposition~\ref{prop:projective}}

Recall the statement of Proposition~\ref{prop:projective}: there is a projective duality between the linear execution trace and the traversal of the static causal tree $\mathcal{T}$. Concretely, for each time index $\tau$ there is a unique pair (leaf, local offset) at which the simulator reconstructs $C_\tau$, and conversely each such pair corresponds to a unique configuration in the linear history.

\begin{proof}[Proof of Proposition~\ref{prop:projective}]
Let the length-$t$ run of $M$ be
\[
\mathcal{L} = (C_0, C_1, \dots, C_t).
\]
Fix a block size $b$ and the associated partition of the time indices into blocks
\[
I_k = \bigl[(k-1)b + 1,\, \min\{kb,t\}\bigr], \quad k = 1,\dots,T,
\]
where $T = \lceil t/b \rceil$. Each time index $\tau \in \{1,\dots,t\}$ lies in a unique block $I_k$.

The height-compression transformation (Theorem~\ref{thm:heightcompression}) produces a balanced computation tree $\mathcal{T}$ whose leaves correspond exactly to the blocks $I_1,\dots,I_T$. We write $\ell_k$ for the leaf node corresponding to block $I_k$.

\medskip
\noindent\textbf{Forward map (history $\to$ traversal).}
For each $\tau \in \{1,\dots,t\}$, define:
\begin{itemize}
    \item $\ell(\tau)$ to be the unique leaf $\ell_k$ such that $\tau \in I_k$, and
    \item $\delta(\tau)$ to be the offset of $\tau$ within $I_k$, for example
    \[
    \delta(\tau) \;=\; \tau - \bigl((k-1)b + 1\bigr),
    \]
    so that $0 \le \delta(\tau) < |I_k|$.
\end{itemize}
This defines a well-posed map
\[
f : \{1,\dots,t\} \to \{\text{leaf nodes of }\mathcal{T}\} \times \{0,\dots,b-1\}, \quad
f(\tau) = \bigl(\ell(\tau), \delta(\tau)\bigr),
\]
with the understanding that only the first $|I_k|$ offsets are used for leaf $\ell_k$ when $I_k$ is shorter than $b$.

\medskip
\noindent\textbf{Traversal behavior of the simulator.}
Consider now the square-root space simulator from Theorem~\ref{thm:sqrtbound}, instantiated with the height-compressed tree $\mathcal{T}$ and a fixed depth-first traversal order $\pi$ on the nodes of $\mathcal{T}$ (for concreteness, a standard pre-order DFS). By construction in~\cite{HeightCompression}:

\begin{enumerate}
    \item The traversal $\pi$ visits each leaf $\ell_k$ exactly once.
    \item When $\pi$ reaches leaf $\ell_k$, the simulator invokes the Algebraic Replay Engine on the summary $\sigma(I_k)$ to regenerate the configurations $C_\tau$ for all $\tau \in I_k$, in increasing order of $\tau$. The replay engine runs as a subroutine whose internal step counter identifies the current offset $\delta$ within $I_k$.
    \item The correctness proof in~\cite{HeightCompression} guarantees that the configurations output during replay coincide exactly with the original configurations in $\mathcal{L}$: when the subroutine reports the configuration at offset $\delta$ within $I_k$, that configuration is equal to $C_\tau$ where $\tau$ is the unique index in $I_k$ with that offset.
\end{enumerate}

Thus, during the execution of the simulator, each pair $(\ell_k,\delta)$ with $0 \le \delta < |I_k|$ is visited exactly once, and at that moment the simulator reconstructs the unique configuration $C_\tau$ such that $f(\tau) = (\ell_k,\delta)$.

This shows that the map
\[
g : \{\text{valid leaf-offset pairs visited along }\pi\} \to \{1,\dots,t\}, \quad
g(\ell_k,\delta) = \tau
\]
is well-defined (by correctness of the simulator) and is the inverse of $f$.

\medskip
\noindent\textbf{Bijection.}
By construction, $f$ and $g$ are inverses of one another:
\[
g\bigl(f(\tau)\bigr) = \tau \quad \text{for all } \tau \in \{1,\dots,t\},
\]
and
\[
f\bigl(g(\ell_k,\delta)\bigr) = (\ell_k,\delta)
\]
for every leaf-offset pair $(\ell_k,\delta)$ actually used during replay. Hence $f$ defines a bijection between the set of time indices and the set of leaf-offset pairs visited during the depth-first traversal $\pi$, and the simulator’s reconstruction of configurations along $\pi$ is projectively equivalent to the linear history $\mathcal{L}$.

This is exactly the claimed projective duality between the linear execution trace and the tree traversal.
\end{proof}


\newpage

\begin{thebibliography}{99}

\bibitem{HopcroftPaulValiant1977}
J.~E.~Hopcroft, W.~J.~Paul, and L.~G.~Valiant.
\newblock On time versus space.
\newblock {\em Journal of the ACM}, 24(2):332--337, 1977.

\bibitem{PaulPippengerSzemerediTrotter1983}
W.~J.~Paul, N.~Pippenger, E.~Szemer{\'e}di, and W.~T.~Trotter.
\newblock On determinism versus nondeterminism and related time-space tradeoffs.
\newblock {\em Proceedings of the 24th IEEE Symposium on Foundations of Computer Science (FOCS)}, 1983.

\bibitem{AroraBarak2009}
S.~Arora and B.~Barak.
\newblock {\em Computational Complexity: A Modern Approach}.
\newblock Cambridge University Press, 2009.

\bibitem{Savitch1970}
W.~J.~Savitch.
\newblock Relationships between nondeterministic and deterministic tape complexities.
\newblock {\em Journal of Computer and System Sciences}, 4(2):177--192, 1970.

\bibitem{Williams2025}
R.~R.~Williams.
\newblock Simulating time with square-root space.
\newblock {\em Proceedings of the 57th Annual ACM Symposium on Theory of Computing (STOC)}, 2025.

\bibitem{HeightCompression}
L.~Nye.
\newblock $\mathrm{TIME}[t]\subseteq \mathrm{SPACE}[O(\sqrt{t})]$ via Tree Height Compression.
\newblock Manuscript, https://arxiv.org/abs/2508.14831, 2025.

\bibitem{LiVitanyi2008}
M.~Li and P.~Vit{\'a}nyi.
\newblock {\em An Introduction to Kolmogorov Complexity and Its Applications}.
\newblock 3rd edition, Springer, 2008.

\bibitem{Bekenstein1973}
J.~D.~Bekenstein.
\newblock Black holes and entropy.
\newblock {\em Physical Review D}, 7(8):2333--2346, 1973.

\bibitem{tHooft1993}
G.~'t~Hooft.
\newblock Dimensional reduction in quantum gravity.
\newblock In {\em Salamfestschrift: A Collection of Talks}, World Scientific, 1993.
\newblock Also available as arXiv:gr-qc/9310026.

\bibitem{Susskind1995}
L.~Susskind.
\newblock The world as a hologram.
\newblock {\em Journal of Mathematical Physics}, 36(11):6377--6396, 1995.

\bibitem{RyuTakayanagi2006}
S.~Ryu and T.~Takayanagi.
\newblock Holographic derivation of entanglement entropy from AdS/CFT.
\newblock {\em Physical Review Letters}, 96(18):181602, 2006.

\bibitem{EisertCramerPlenio2010}
J.~Eisert, M.~Cramer, and M.~B.~Plenio.
\newblock Colloquium: Area laws for the entanglement entropy.
\newblock {\em Reviews of Modern Physics}, 82(1):277--306, 2010.

\end{thebibliography}
\end{document}